\documentclass[a4paper]{article}  
\usepackage{amsmath}
\usepackage{amsfonts,amssymb}
\usepackage{amsthm}
\usepackage{dsfont}
\usepackage{psfig}

\begin{document}

\title{Exact mapping between system-reservoir quantum models and semi-infinite discrete chains using orthogonal polynomials}

\author{\small Alex W. Chin$^{1}$\footnote{alex.chin@uni-ulm.de}, \'Angel Rivas$^{1}$, Susana F. Huelga$^{1}$ and Martin B. Plenio$^{1,2}$}
\date{\small {\it $^{1}$Institut f\"{u}r Theoretische Physik, Universit\"{a}t Ulm, Ulm D-89069, Germany. \\
  $^{2}$Institute for Mathematical Sciences, Imperial College London, London SW7 2PG, United Kingdom.}\\
  (Dated: June 21, 2010)}

\maketitle
\begin{abstract}
By using the properties of orthogonal polynomials, we present an exact unitary transformation that maps the Hamiltonian of a quantum system coupled linearly to a continuum of bosonic or fermionic modes to a Hamiltonian that describes a one-dimensional chain with only nearest-neighbour interactions. This analytical transformation predicts a simple set of relations between the parameters of the chain and the recurrence coefficients of the orthogonal polynomials used in the transformation, and allows the chain parameters to be computed using numerically stable algorithms that have been developed to compute recurrence coefficients. We then prove some general properties of this chain system for a wide range of spectral functions, and give examples drawn from physical systems where exact analytic expressions for the chain properties can be obtained. Crucially, the short range interactions of the effective chain system permits these open quantum systems to be efficiently simulated by the density matrix renormalization group methods.
\end{abstract}

\section{Introduction}
\newtheorem{theorem}{Theorem}[section]
\newtheorem{proposition}{Proposition}[section]
\newtheorem{collolary}{Collolary}[section]
\theoremstyle{definition}
\newtheorem{definition}{Definition}[section]
\newtheorem{remark}{Remark}[section]
\newtheorem{corollary}{Corollary}[section]

All quantum systems encountered in nature experience random perturbations due to their coupling to degrees of freedom of their local environments. Information about these degrees of freedom are not normally accessible, and to correctly predict the results of experiments, these degrees of freedom must be averaged over. This averaging introduces qualitatively new features into the otherwise unitary dynamics of the quantum system, and typically induces an effectively irreversible dynamics which drives the system towards an equilibrium with its environment \cite{weiss}. For quantum systems, this evolution towards equilibrium not only involves energy transfer to the environment, but can also cause the loss of coherence in the system state, often on a much faster timescale than the energy relaxation \cite{weiss,petru2}. This latter process of decoherence rapidly destroys quantum mechanical effects arising from the existence of phase coherence in the state of the system, and must always be considered when analyzing the results of experimental investigations of quantum phenomena. Decoherence also presents a major problem for the emerging field of quantum information and technology, a branch of physics that aims to directly harness quantum mechanical effects to create novel and highly effective devices which could greatly outperform their classical counterparts. This latter issue has attracted considerable interest in recent years, and the study of quantum systems in contact with environments, often referred to as `open-quantum systems', is a key problem in contemporary physics \cite{petru2}.

The theoretical description of a quantum system coupled to environmental degrees of freedom has been developed over many decades, and of particular important in these studies is the spin-boson model (SBM) \cite{weiss,leggett}, one of the simplest non-trivial models of open-system dynamics. This model describes a single two-level system (TLS) coupled linearly to the coordinates of an environment consisting of a continuum of harmonic oscillators. Despite its simplicity, this model cannot be solved exactly and shows a rich array of non-Markovian dynamical phenomena, including a type of quantum phase transition between dynamically localized and delocalized states of the TLS for so-called Ohmic and sub-Ohmic baths at zero temperature \cite{leggett,chin05,kehrein96,bulla,spohn,winter}. These interesting dynamics normally appear in situations where the effective interaction between the TLS and the oscillators cannot be treated perturbatively, and although many analytical and formal approaches have shed considerable light on the nature of the effects, the explicit time evolution of the TLS in these regimes must normally be simulated numerically. These simulations can be loosely divided into methods in which the environmental degrees of freedom are eliminated in order to derive an effective equation of motion for the TLS, and those in which the complete many-body dynamics of the TLS and the environment are computed. The former approach includes the hierarchy approach and various numerical path integral techniques, while the latter includes exact diagonalization and the numerical reorganization group (NRG) approaches \cite{numerics}.

Recently, the time-adaptive density matrix renormalisation group (t-DMRG) \cite{white92} has also been added to this category. This approach provides numerically-exact simulations for the dynamics of one-dimensional systems with short-range interactions, and is regarded as one of the most powerful and versatile methods in condensed matter, atomic, and optical physics for the study of strongly-correlated many-body quantum states. The recent study of Prior {\it et al.} \cite{dimer} successfully used a t-DMRG approach for a two-spin system where each spin was in contact with an oscillator environment with arbitrarily complex spectral densities. The implementation and accuracy of this approach rests on an exact mapping between the canonical SBM and a one-dimensional harmonic chain with short range interactions. The reasons why this form of Hamiltonian allows the t-DMRG to efficiently simulate extremely large chains is disscussed at length in Ref. \cite{white92}.

The idea of this mapping was first introduced in the NRG studies of Bulla {\it et al.}~\cite{bulla}, who first discretise the continuous environment and then perform a recursive numerical technique to bring the truncated finite Hamiltonian to the desired chain form. The dynamics of this chain can then be numiercally simulated \cite{anders07}. In this paper we show how the formal properties of orthogonal polynomials can be used to implement an exact unitary transformation which maps an arbitrary continuous SBM directly onto the nearest-neighbour $1-D$ chain without any need for discretisation. In several important cases we can actually perform the mapping analytically and predict the energies and couplings of the chain in closed form. Our analyisis of the mapping shows that these energies and couplings are simply related to the recurrence coefficients of orthogonal polynomials, and when analytical results are not available, we can make use of sophisticated numerical techniques that compute recurrence coefficients with high precision and which do not suffer from the numerical instabilities which often appear in the standard recursive technique. Not only does this transformation aid the simulation of non-perturbative open-system dynamics, we shall show that the use of orthogonal polynomials also reveals some universal features of TLS dynamics which are independent of the spectral density of the environmental oscillators, and which may point the way to even more efficient simulations of strongly non-markovian open-system dynamics.

The plan of this paper is as follows. In section \ref{sec:orthogonalpolynomials} we introduce the formal properties of orthogonal polynomials needed to implement the chain mapping, and in section \ref{sec;system-reservoir} we introduce the general model of an open-quantum system  which can be mapped onto a $1-D$ chain suitable for t-DMRG simulation. This section also contains details about the transformation itself and analytical expressions for the frequencies and nearest-neighbour interactions of the chain are presented and discussed. Section \ref{sec;examples1} gives explicit examples of this transformation for the SBM, and \ref{sec;discrete} deals with an extension of the orthogonal polynomial method to systems interacting with a discrete number of environmental modes. These discrete environments may be found in certain physical systems, and are also important in the NRG description of open-quantum systems where the continuous spectrum of the environment is approximated by a discrete set of oscillators \cite{bulla}. We then end with some general conclusions and a brief recapitulation of the main results.

\section{Orthogonal Polynomials}
\label{sec:orthogonalpolynomials}
In this section we present a few useful properties of orthogonal polynomials that will be employed in the following. This section also serves to fix our notation.
\subsection{Basic concepts}
Let $\mathbb{P}$ be the space of the polynomials with real coefficients (the following may also be extended to complex, but real coefficients will be enough for our proposes), and $\mathbb{P}_n\subset\mathbb{P}$ the space of polynomials of degree equal to $n$, this is $p_n(x)\in\mathbb{P}_n$:
\[
p_n(x)=\sum_{m=0}^n a_mx^m,
\]
where the $a_m\in\mathds{R}$ and $x$ is a real variable.

\begin{definition}
A polynomial $p_n(x)\in\mathbb{P}_n$ is called \emph{monic} if its leading coefficient satisfies $a_n=1$, will be denoted by $p_n(x)\equiv\pi_n(x)$. Of course every polynomial of $\mathbb{P}_n$ can become monic just dividing it by $a_n$.
\end{definition}

\begin{definition}
Let $\mu(x)$ be a nondecreasing and differentiable function on some interval $[a,b]\in\mathds{R}$, and assume that the induced positive measure $d\mu(x)$ has finite moments of all orders:
\[
\int_a^bx^r d\mu(x)<\infty, \quad r=0,1,2,\ldots
\]
Then for any $p(x),q(x)\in\mathbb{P}$ we define an \emph{inner product} as
\[
\langle p,q\rangle_\mu=\int_a^bp(x)q(x)d\mu(x).
\]
\end{definition}
Of course it defines a \emph{norm} via
\[
\|p\|_\mu=\sqrt{\langle p, p\rangle_\mu}=\left(\int_a^bp^2(x)d\mu(x)\right)^{1/2}\geq0,
\]
and satisfies the Cauchy-Schwarz's inequality,
\[
|\langle p,q\rangle_\mu |\leq \|p\|_\mu \|q\|_\mu.
\]
\begin{definition}
Two polynomials $p(x),q(x)\in\mathbb{P}$ are said to be \emph{orthogonal} with respect to the measure $d\mu(x)$ if its inner product is zero $\langle p,q\rangle_\mu=0$.
\end{definition}
\begin{definition}
A set of monic polynomials $\{\pi_n(x)\in \mathbb{P}_n,n=0,1,2,\ldots\}$ is called an \emph{orthogonal set} with respect to the measure $d\mu(x)$ if
\[
\langle \pi_n, \pi_m\rangle_\mu=\|\pi_n\|_\mu^2\delta_{n,m} \quad\text{for }n,m=0,1,2,\ldots
\]
Sometimes it is useful to normalize the monic polynomials, $\tilde{p}_n(x)=\pi_n(x)/\|\pi_n(x)\|_\mu$, such that $\langle \tilde{p}_n, \tilde{p}_m\rangle_\mu=\delta_{n,m}$, then the set is called \emph{orthonormal}, however note that in general the normalized polynomials lose the monic character.
\end{definition}

\begin{theorem} For any (non-zero) measure $d\mu(x)$ there exist a unique family of monic orthogonal polynomials $\{\pi_n(x)\in \mathbb{P}_n,n=0,1,2,\ldots\}$.
\end{theorem}
\begin{proof} The proof is by construction. Applying the well-known {\it Gram-Schmidt orthogonalization procedure} to the sequence $\{1,x,x^2,\ldots\}$, the monic orthogonal polynomials are constructed from $\pi_{0}=1$ using the formula:
\[
\pi_k(x)=m_k(x)-\sum_{n=0}^{k-1}\left(\frac{\langle m_k,\pi_n\rangle_\mu}{\langle \pi_n,\pi_n\rangle_\mu}\right)\pi_n
\]
with $m_k(x)=x^k$.
\end{proof}

\begin{theorem} \label{theorem2} The set of monic polynomials $\{\pi_n(x)\in \mathbb{P}_n,n=0,1,2,\ldots\}$ is a basis of the space $\mathbb{P}$. More specifically if $p(x)$ is a polynomial of degree less or equal to $r$, then it can be uniquely represented by
\begin{equation}\label{linearspan}
p(x)=\sum_{n=0}^r c_n\pi_n(x)
\end{equation}
for some real constants $c_k$.
\end{theorem}
\begin{proof}
The proof is by mathematical induction. If $r=0$ the proof is trivial. Assume that the statement is true for $r-1$, then there is a unique $c_r$ such that $p(x)-c_r\pi_r(x)=p_{r-1}(x)$ is polynomial with degree less or equal to $r-1$. Moreover, if the set of monic polynomials is orthogonal, taking the inner product on (\ref{linearspan}) with respect to $\pi_m(x)$ we obtain
\[
c_n=\frac{\langle \pi_n,p\rangle_\mu}{\|\pi_n\|_\mu^2}.\quad
\]
\end{proof}

There are of course many more properties like these which are similar to those found in linear algebra spaces;
however we are not going to go further in this way here, though the reader can find more information in references \cite{gautschi,Chihara,Szego} for example. From now on we omit the subindex $\mu$ in the inner product and norms in order to make the notation less dense.

\subsection{Recurrence Relations}
One of the most useful properties for our interests is the recurrence formula of the orthogonal polynomials. Here we divide the section into two cases, monic and orthonormal polynomials.

\begin{theorem} Let $\{\pi_n(x)\in \mathbb{P}_n,n=0,1,2,\ldots\}$ be the set of monic orthogonal polynomials with respect to some measure $d\mu(x)$, then
\begin{equation}\label{monicrecurrence}
\pi_{k+1}(x)=(x-\alpha_k)\pi_k(x)-\beta_k\pi_{k-1}(x), \quad k=0,1,2...
\end{equation}
where $\pi_{-1}(x)\equiv0$, and the recurrence coefficients are:
\begin{equation}\label{akbk}
\alpha_k=\frac{\langle x \pi_k,\pi_k\rangle}{\langle \pi_k,\pi_k\rangle}, \quad \beta_k=\frac{\langle \pi_k,\pi_k\rangle}{\langle \pi_{k-1},\pi_{k-1}\rangle}.
\end{equation}
\end{theorem}
\begin{proof} By theorem \ref{theorem2}, $\pi_{k+1}(x)-x\pi_k(x)$ is a polynomial of degree less or equal to $k$, we can write it in terms of monic polynomials of degree less or equal $k$:
\[
\pi_{k+1}(x)-x\pi_k(x)=\sum_{n=0}^k c_n \pi_n(x),
\]
and since the set is orthogonal
\[
c_k=\frac{\langle\pi_{k+1}-x\pi_k,\pi_k\rangle}{\langle \pi_k,\pi_k\rangle}=-\frac{\langle x\pi_k,\pi_k\rangle}{\langle \pi_k,\pi_k\rangle},
\]
similarly
\[
c_{k-1}=\frac{\langle\pi_{k+1}-x\pi_k,\pi_{k-1}\rangle}{\langle \pi_{k-1},\pi_{k-1}\rangle}=-\frac{\langle x\pi_k,\pi_{k-1}\rangle}{\langle \pi_{k-1},\pi_{k-1}\rangle}=-\frac{\langle \pi_k,x\pi_{k-1}\rangle}{\langle \pi_{k-1},\pi_{k-1}\rangle}=-\frac{\langle \pi_k,\pi_{k}\rangle}{\langle \pi_{k-1},\pi_{k-1}\rangle},
\]
where in the last step we have used the fact that terms of degree less than $k$ are orthogonal to $\pi_k(x)$, because they can be expressed as linear combination of $\{\pi_n(x)\in \mathbb{P}_n,n=0,1,2,\ldots,k-1\}$ which are all orthogonal to $\pi_k(x)$. By the same argument, it is easy to see that the remaining coefficients are zero $c_{n}=0$ for $n=0,\ldots,k-2$. So by comparing with equation (\ref{monicrecurrence}) we obtain $\alpha_k=-c_k$ and $\beta_k=-c_{k-1}$.
\end{proof}

\begin{theorem}
Assume that we have a set of orthonormal polynomials $\{\tilde{p}_n(x)\in \mathbb{P}_n,n=0,1,2,\ldots\}$ then, the recurrence relation is:
\begin{equation}\label{normalrecurrence}
\tilde{p}_{k+1}(x)=(C_kx-A_k)\tilde{p}_k(x)-B_k\tilde{p}_{k-1}(x), \quad k=0,1,2...
\end{equation}
where $\tilde{p}_{-1}(x)\equiv0$. If we denote by $\kappa_n$ the leading coefficient of $\tilde{p}_n(x)$ and by $\kappa'_n$ the second leading coefficient, i.e. $\tilde{p}_n(x)=\kappa_nx^n+\kappa'_nx^{n-1}+\ldots$, we get
\begin{equation}
A_k=\frac{\kappa_{k+1}}{\kappa_{k}}\left(\frac{\kappa'_{k}}{\kappa_{k}}-\frac{\kappa'_{k+1}}{\kappa_{k+1}}\right), \quad B_k=\frac{\kappa_{k+1}\kappa_{k-1}}{\kappa^2_{k}}, \quad C_k=\frac{\kappa_{k+1}}{\kappa_{k}} \label{AkBkCk}
\end{equation}
\end{theorem}
\begin{proof} The proof is not more complicated than in the monic case, if we chose $C_k=\frac{\kappa_{k+1}}{\kappa_{k}}$ then, $\tilde{p}_{k+1}(x)-C_kx\tilde{p}_k(x)$ is a polynomial of degree lower of equal to $k$, in the same fashion the choice
\[
A_k=\frac{C_k\kappa'_{k}-\kappa'_{k+1}}{\kappa_k}=\frac{\kappa_{k+1}}{\kappa_{k}}\left(\frac{\kappa'_{k}}{\kappa_{k}}-\frac{\kappa'_{k+1}}{\kappa_{k+1}}\right)
\]
makes $\tilde{p}_{k+1}(x)-(C_kx-A_k)\tilde{p}_k(x)$ a polynomial of degree lower of equal to $k-1$. By orthogonality the inner product of it with $\tilde{p}_n$ for $n\leq k-2$ are zero, and the inner product with $\tilde{p}_{k-1}$ yields
\begin{eqnarray*}
\langle\tilde{p}_{k-1},\tilde{p}_{k+1}-(C_kx\tilde{p}_k+A_k)\tilde{p}_k\rangle=-C_k\langle\tilde{p}_{k-1},x\tilde{p}_k\rangle=-C_k\langle x\tilde{p}_{k-1},\tilde{p}_k\rangle\\
=-C_k \frac{\kappa_{k-1}}{\kappa_{k}}\langle \tilde{p}_{k},\tilde{p}_k\rangle=-C_k \frac{\kappa_{k-1}}{\kappa_{k}},
\end{eqnarray*}
so by taking into account that this is the coefficient which goes with $\tilde{p}_{k-1}$, replacing the value of $C_k$ gives the value of $B_k$ in (\ref{AkBkCk}).
\end{proof}

The relation between both cases is given by the next proposition.

\begin{proposition} \label{proposition1} Let $\{\pi_n(x)\in \mathbb{P}_n,n=0,1,2,\ldots\}$ a set of monic orthogonal polynomials and $\tilde{p}_n(x)=\pi_n(x)/\|\pi_n(x)\|$ their orthonomalized version, then the coefficients for the respective recurrence relations are related by:
\begin{equation}\label{monic-normal}
A_k=\frac{\alpha_k}{\sqrt{\beta_{k+1}}}, \quad B_k=\sqrt{\frac{\beta_k}{\beta_{k+1}}} \quad\text{and } C_k=\frac{1}{\sqrt{\beta_{k+1}}}.
\end{equation}
\end{proposition}
\begin{proof} Inserting $\pi_n(x)=\|\pi_n(x)\|\tilde{p}_n(x)$ in (\ref{monicrecurrence}) and dividing by $\|\pi_{k+1}\|$ one obtains
\[
\tilde{p}_{k+1}(x)=\frac{\|\pi_k(x)\|}{\|\pi_{k+1}(x)\|}(x-\alpha_k)\tilde{p}_k(x)-\frac{\|\pi_{k-1}(x)\|}{\|\pi_{k+1}(x)\|}\tilde{p}_{k-1}(x).
\]
Employing $\beta_k=\|\pi_k\|^2/\|\pi_{k-1}\|^2$ and comparing with the above (\ref{normalrecurrence}) we obtain (\ref{monic-normal}).
\end{proof}

\subsection{Boundness properties of the recurrence coefficients} A particular important property for our work will be the behaviour of the monic recurrence coefficients $\alpha_n$ and $\beta_n$ when $n$ is increasing. To start, we have the following theorem.
\begin{theorem} Let the support interval $[a,b]$ of $d\mu(x)$ be finite, then for each $k=0,1,2,\ldots$
\begin{eqnarray*}
&a\leq \alpha_k \leq b,\\
&\beta_k\leq\max\{a^2,b^2\}.
\end{eqnarray*}
\end{theorem}
\begin{proof} The first relation follows immediately from $\alpha_k=\frac{\langle x \pi_k,\pi_k\rangle}{\langle \pi_k,\pi_k\rangle}$ because $\pi^2_k(x)$ is positive for every $x$ and $a\leq x\leq b$ inside of the integration domain. For the second relation we have by orthogonality
\[
\langle \pi_k,\pi_k\rangle=\langle x\pi_{k-1},\pi_k\rangle,
\]
and the Cauchy-Schwarz's inequality gives
\[
\langle \pi_k,\pi_k\rangle\leq |x|\|\pi_{k-1}\|\|\pi_k\|\leq \max\{|a|,|b|\}\|\pi_{k-1}\|\|\pi_k\|,
\]
finally the division of both sides by $\|\pi_{k}\|$ and $\|\pi_{k-1}\|$ and squaring leads straightforwardly to the second bound.
\end{proof}

Note that for measures with unbounded support, the monic recurrence coefficients are not bounded as $n$ increases; see for example the case of Laguerre and Hermite polynomials \cite{Handbook}. However, for measures with bounded support further useful properties about the recurrence coefficients can be obtained, to do this we need the following results.

\begin{proposition} Let $\{\pi_n(x)\in \mathbb{P}_n,n=0,1,2,\ldots\}$ be a set of monic orthonormal polynomials with respect to some measure $d\mu(x)$ supported in $[-1,1]$, with recursion coefficients $\alpha_k$ and $\beta_k$ defined according to (\ref{akbk}). If we change the support of the measure to $[a,b]$ by means of a linear transformation $y=m x+c$, where
\begin{equation}
m=\frac{b-a}{2},\quad c=\frac{a+b}{2},
\end{equation}
the recursion coefficients $\alpha'_k$ and $\beta'_k$ of the transformed monic polynomials $\{\pi'_n(y)$ $\in \mathbb{P}_n,n=0,1,2,\ldots\}$ are given by
\begin{equation} \label{a'kb'k}
\alpha'_k=m \alpha_k+c, \quad \beta'_k=m^2\beta_k.
\end{equation}
\end{proposition}
\begin{proof} The map $y=m x+c$ transform the measure $d\mu(y)=m d\mu(x)$ and every monic polynomial $\pi_k(x)=\pi_k(\frac{y-c}{m})$ which is not a monic polynomial in $y$. To make these polynomials monic in $y$,
we need to multiply by $m^k$, that is, the set of transformed monic polynomials is $\{\pi'_n(y)=m^k\pi_k(x)\in \mathbb{P}_n,n=0,1,2,\ldots\}$. Then we apply the definitions (\ref{akbk})
\begin{eqnarray*}
\alpha'_k=\frac{\langle y \pi'_k,\pi'_k\rangle}{\langle \pi'_k,\pi'_k\rangle}&=&\frac{\left(m \langle  x \pi'_k,\pi'_k\rangle+c\langle\pi'_k,\pi'_k\rangle\right)}{\langle \pi'_k,\pi'_k\rangle}\\
&=&\frac{m^{2k+1}\left(m\langle  x \pi_k,\pi_k\rangle+c\langle\pi_k,\pi_k\rangle\right)}{m^{2k+1} \langle \pi_k,\pi_k\rangle}=m \alpha_k+c,
\end{eqnarray*}
and similarly
\[
\beta'_k=\frac{\langle \pi'_k,\pi'_k\rangle}{\langle \pi'_{k-1},\pi'_{k-1}\rangle}=\frac{m^{2k+1}\langle \pi_k,\pi_k\rangle}{m^{2(k-1)+1}\langle \pi_{k-1},\pi_{k-1}\rangle}=m^2\beta_k.\quad
\]
\end{proof}

\begin{definition} A measure $d\mu(x)=w(x)dx$ supported in $[-1,1]$ belongs to the Szeg\"o class if
\begin{equation}\label{SzegoCondition}
\int_{-1}^1\frac{\ln w(x)}{\sqrt{1-x^2}}dx>-\infty.
\end{equation}
\end{definition}

\begin{theorem}[Szeg\"o] Let $\{\tilde{p}_n(x)\in \mathbb{P}_n,n=0,1,2,\ldots\}$ orthonormal polynomials with respect to some Szeg\"o class measure $d\mu(x)=w(x)dx$, then in the asymptotic limit $n\rightarrow\infty$ we have the approximate expressions for the first leading coefficient $\kappa_n$ and for the second one $\kappa'_n$,
\begin{eqnarray}
\kappa_n&\simeq& \frac{2^n}{\sqrt{\pi}}\exp\left[-\frac{1}{2\pi}\int_{-1}^1\frac{\ln w(x)}{\sqrt{1-x^2}}dx\right],\label{Knasin}\\
\kappa'_n&\simeq& -\frac{2^{n-1}}{\pi^{3/2}}\left[\int_{-1}^1\frac{x\ln w(x)}{\sqrt{1-x^2}}dx\right]\exp\left[-\frac{1}{2\pi}\int_{-1}^1\frac{\ln w(x)}{\sqrt{1-x^2}}dx\right]\label{Kn'asin}.
\end{eqnarray}
\end{theorem}
\begin{proof} The proof can be founded in chapter XII of \cite{Szego}.
\end{proof}

\begin{proposition} The coefficients $A_n$, $B_n$ and $C_n$ have the asymptotic values
\begin{equation}
\lim_{n\rightarrow\infty}A_n=0,\quad \lim_{n\rightarrow\infty}B_n=1, \quad \lim_{n\rightarrow\infty}C_n=2.
\end{equation}
\end{proposition}
\begin{proof}
It follows immediately from (\ref{AkBkCk}) and the previous result (\ref{Knasin}) and (\ref{Kn'asin}).
\end{proof}
The speed of convergence of $A_n$, $B_n$ and $C_n$ depends on the measure in a non-trivial fashion (see details \cite{Szego}). We will demonstrate this explicitly in equations (\ref{exenergies}) and (\ref{extn}) for specific spectral densities.

\begin{corollary} Let $\{\pi_n(x)\in \mathbb{P}_n,n=0,1,2,\ldots\}$ be a set of orthogonal monic polynomials with respect to some measure $d\mu(x)$ which belongs to the Szeg\"o class, then their recurrence coefficients have the asymptotic values
\[
\lim_{n\rightarrow\infty}\alpha_n=0, \quad \lim_{n\rightarrow\infty}\beta_n=\frac{1}{4}.
\]
\end{corollary}
\begin{proof} See proposition \ref{proposition1}.
\end{proof}

\begin{proposition} Let $\{\pi_n(x)\in \mathbb{P}_n,n=0,1,2,\ldots\}$ be a set oforthogonal monic polynomials with respect to some Szeg\"o class measure, then asymptotic values of the monic recursion coefficients on any other finite support $[a,b]$ are
\[
\lim_{n\rightarrow\infty}\alpha'_n=c=\frac{a+b}{2}, \quad \lim_{n\rightarrow\infty}\beta'_n=\frac{m^2}{4}=\frac{(b-a)^2}{16}.
\]
\end{proposition}
\begin{proof} It is evident that if (\ref{SzegoCondition}) is satisfied, the linear transformation $x\rightarrow m x+c$ which changes the measure support to $[a,b]$ does not affect to its fulfilment since it is not singular in the interval, the rest is a simple application of the result (\ref{a'kb'k}).
\end{proof}

\section{System-Reservoir structures}
\label{sec;system-reservoir}

The observables in an open quantum system are affected by the unavoidable interaction with the environment. This environment may be described by an infinite number (often called a reservoir) of bosonic or femionic modes labeled by some real number. The internal dynamics of the reservoir is given by some Hamiltonian of the form
\[
H_{\mathrm{res}}=\int_{0}^{x_{\mathrm{max}}}dx g(x)a_{x}^{\dagger}a_{x},
\]
where in a physical context $x$ could represent some continuous real variable such as the momentum of each mode, and $x_{\mathrm{max}}$ the maximum value of it which is present in the reservoir (it could be infinity). In this picture $g(x)$ represents the dispersion relation of the reservoir which relates the oscillator frequency to the variable $x$. The creation and annihilation operators satisfy the continuum bosonic $[a_{x},a_{y}^{\dagger}]=\delta(x-y)$ or fermionic $\{a_{x},a_{y}^{\dagger}\}=\delta(x-y)$ commutation rules. We assume that the frequencies $g(x)$ and momenta $x$ of the reservoir are bounded.

The internal dynamics of the open quantum system are described by a unspecified local Hamiltonian operator $H_{\mathrm{loc}}$ and we assume that the interaction between the system and the reservoir is given by a linear coupling
\[
V=\int_{0}^{x_{\mathrm{max}}}dxh(x)\hat{A}(a_{x}+a^{\dagger}_{x}),
\]
where $\hat{A}$ is some operator of the system and the coupling (real) function $h(x)$ describes the coupling strength with each mode. The total Hamiltonian for system plus reservoir is then given by
\begin{equation}\label{totalHamiltonian}
H=H_{\mathrm{loc}}+H_{\mathrm{res}}+V=H_{\mathrm{loc}}+\int_{0}^{x_{\mathrm{max}}}dx g(x)a_{x}^{\dagger}a_{x}+\int_{0}^{x_{\mathrm{max}}}dxh(x)\hat{A}(a_{x}+a_{x}^{\dagger}).
\end{equation}

It has been shown that the dynamics induced in the quantum system by its interaction with the reservoir is completely determined by a positive function of the energy (or frequency $\omega$) of the oscillators called the spectral density $J(\omega)$ \cite{weiss,leggett}. For the continuum model of the reservoir we are considering, this function is given by,
\begin{equation}\label{Jhg}
J(\omega)=\pi h^2[g^{-1}(\omega)]\frac{dg^{-1}(\omega)}{d\omega},
\end{equation}
where $g^{-1}[g(x)]=g[g^{-1}(x)]=x$. Physically, $\frac{d g^{-1}(\omega)}{d\omega} \delta\omega$, can be interpreted as the number of quanta with frequencies between $\omega$ and $\omega+\delta\omega$ as $\delta\omega\rightarrow 0$ i.e. it represents the density of states of the reservoir in frequency space. The spectral function thus describes the overall strength of the interaction of the system with the reservoir modes of frequency $\omega$.

This physical introduction motivate the next mathematical definition.

\begin{definition} We define the spectral density of a reservoir as a real non-negative and integrable
function $J(\omega)$ inside of its (real positive) domain, which could be the entire half-line $\omega\in[0,\infty)$.
\end{definition}

Of course, given only a spectral density $J(\omega)$, the dispersion relation $g(x)$ and the coupling function $h(x)$ are not uniquely defined. We shall make use of this freedom to implement a particularly simple transformation of the bosonic modes, and we will chose the dispersion function to be linear $g(x)=gx$. We now move on to our main theorem.

\begin{theorem}\label{mainResult} A system linearly coupled with a reservoir characterized by a spectral density $J(\omega)$ is unitarily equivalent to semi-infinite chain with only nearest-neighbors interactions, where the system only couples to the first site in the chain (see figure \ref{fig1}). In other words, there exist an unitary operator $U_n(x)$ such that the countably infinite set of new operators
\begin{equation}\label{bn}
b_{n}^{\dagger}=\int_{0}^{x_{\mathrm{max}}}dx\,U_{n}(x)a_{x}^{\dagger},
\end{equation}
satisfy the corresponding commutation relations $[b_n,b^\dagger_m]=\delta_{nm}$ for bosons, and $\{b_n,b^\dagger_m\}=\delta_{nm}$ for fermions, with transformed Hamiltonian
\begin{equation}
H'=\int_{0}^{x_{\mathrm{max}}}dx\,U_{n}(x)H=H_{\mathrm{loc}}+c_0\hat{A}(b_{0}+b_{0}^{\dagger})+\sum_{n=0}^\infty\omega_{n}b_{n}^{\dagger}b_{n}+t_{n}b_{n+1}^{\dagger}b_{n}+t_{n}b_{n}^{\dagger}b_{n+1}\label{hc},
\end{equation}
where $c_0$, $t_{n}$, $\omega_{n}$ are real constants.
\end{theorem}
\begin{proof} The proof is by construction. Since $J(\omega)$ is positive, $h(x)$ is real, this defines the measure $d\mu(x)=h^2(x)dx$. Then write
\begin{equation}\label{transf}
U_n(x)=h(x)\tilde{p}_n(x)=h(x)\frac{\pi_n(x)}{\|\pi_n\|},
\end{equation}
where $\tilde{p}_n(x)$ are some set of orthonormal polynomials with respect to the measure $d\mu(x)=h^2(x)dx$ with support on $[0,x_{\mathrm{max}}]$. Then it is clear that $U_n(x)$ is unitary (actually orthogonal as it is also a real transformation) in the sense of
\begin{eqnarray*}
\int_0^{x_{\mathrm{max}}}dxU_n(x)U_m^\ast(x)&=&\int_0^{x_{\mathrm{max}}}dxU_n(x)U_m(x)\\
&=&\int_0^{x_{\mathrm{max}}}dxh^2(x)\tilde{p}_n(x)\tilde{p}_m(x)=\delta_{nm},
\end{eqnarray*}
so the inverse transformation is just
\begin{equation}\label{ax}
a_{x}^{\dagger}=\sum_{n}U_{n}(x)b_{n}^{\dagger}.
\end{equation}
Moreover, for bosons
\begin{eqnarray*}
[b_n,b_{m}^{\dagger}]&=&\int_{0}^{x_{\mathrm{max}}}\int_{0}^{x_{\mathrm{max}}}dxdx'U_{n}(x)U_{m}(x')[a_x,a_{x'}^{\dagger}]\\
&=&\int_{0}^{x_{\mathrm{max}}}\int_{0}^{x_{\mathrm{max}}}dxdx'U_{n}(x)U_{m}(x')\delta(x-x')\\
&=&\int_{0}^{x_{\mathrm{max}}}dxU_{n}(x)U_{m}(x)=\delta_{nm}.
\end{eqnarray*}
and similarly it is proved that $\{b_n,b^\dagger_m\}=\delta_{nm}$ for fermions. It remains to determine the structure of the transformed Hamiltonian $H'$, note that $V$ is transformed like:
\[
V=\hat{A}\sum_n\int_{0}^{x_{\mathrm{max}}}dxh(x)U_{n}(x)(b_{n}+b_{n}^{\dagger})=\hat{A}\sum_n\int_{0}^{x_{\mathrm{max}}}dxh^2(x)\frac{\pi_n(x)}{\|\pi_n\|}(b_{n}+b_{n}^{\dagger}),
\]
since for monic polynomials $\pi_0(x)=1$ we find
\begin{eqnarray*}
V&=&\hat{A}\sum_n\int_{0}^{x_{\mathrm{max}}}dxh^2(x)\frac{\pi_n(x)\pi_0(x)}{\|\pi_n\|}(b_{n}+b_{n}^{\dagger})\\
&=&\hat{A}\sum_n\frac{\|\pi_n\|^2\delta_{n0}}{\|\pi_n\|}(b_{n}+b_{n}^{\dagger})=\|\pi_0\|\hat{A}(b_{0}+b_{0}^{\dagger}),
\end{eqnarray*}
so $c_0=\|\pi_0\|$. For the $H_{\mathrm{res}}$ term, note that with the choice of linear dispersion function $g(x)=gx$ for the spectral density $J(\omega)$, one obtains
\begin{eqnarray*}
H_{\mathrm{res}}&=&\sum_{n,m}\int_{0}^{x_{\mathrm{max}}}dx g(x)U_{n}(x)U_{m}(x)b_{n}^{\dagger}b_{m}\\
&=&\sum_{n,m}\int_{0}^{x_{\mathrm{max}}}dx h^2(x) g(x)\tilde{p}_n(x)\tilde{p}_m(x) b_{n}^{\dagger}b_{m}\\
&=&g\sum_{n,m}\int_{0}^{x_{\mathrm{max}}}dx x h^2(x)\tilde{p}_n(x)\tilde{p}_m(x)b_{n}^{\dagger}b_{m}.
\end{eqnarray*}
With the recurrence relation (\ref{normalrecurrence}) we substitute the value of $x \tilde{p}_n(x)$ in the above integral to find
\[
H'_{\mathrm{res}}=g\sum_{n,m}\int_{0}^{x_{\mathrm{max}}}dx h^2(x)\left[\frac{1}{C_n}\tilde{p}_{n+1}(x)+\frac{A_n}{C_n}\tilde{p}_{n}(x)+\frac{B_n}{C_n}\tilde{p}_{n-1}(x)\right]\tilde{p}_m(x)b_{n}^{\dagger}b_{m}
\]
then orthonormality yields
\begin{eqnarray*}
H'_{\mathrm{res}}&=&g\sum_{n}\frac{1}{C_n}b_{n}^{\dagger}b_{n+1}+\frac{A_n}{C_n}b_{n}^{\dagger}b_{n}+\frac{B_{n+1}}{C_{n+1}}b_{n+1}^{\dagger}b_{n}\\
&=&g\sum_{n}\sqrt{\beta_{n+1}}b_{n}^{\dagger}b_{n+1}+\alpha_nb_{n}^{\dagger}b_{n}+\sqrt{\beta_{n+1}}b_{n+1}^{\dagger}b_{n}
\end{eqnarray*}
where we have used the relation between monic and orthogomal recurrence coefficients (\ref{monic-normal}). So finally we have
\begin{equation*}
\omega_n=g\alpha_{n},
\end{equation*}
\begin{equation*}
t_n=g\sqrt{\beta_{n+1}}.
\end{equation*}
\end{proof}

\begin{figure}
\begin{center}
\psfig{file=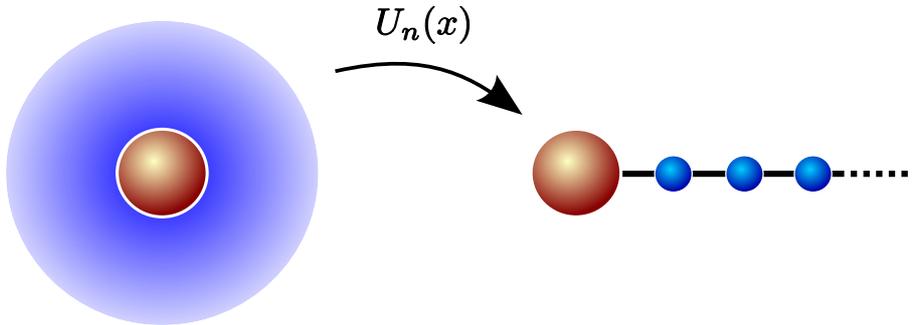,width=\textwidth}
\end{center}
\caption{Illustration of the effect of the transformation $U_n(x)$. On the left side system central is interacting with a reservoir with continuous bosonic or fermionic modes, paremetrized by $x$, after $U_n(x)$ the system is the first place of a discrete semi-infinite chain parametrized by $n$.}
\label{fig1}
\end{figure}

\begin{remark}Note that from a topological point of view, mapping a continuous reservoir into a discrete chain is nothing but accounting for the separability of the topological space $\mathbb{P}$ (or the continuous functions space, if one prefers), and therefore the existence of a numerable basis.
\end{remark}

\begin{remark}In addition, the theorem provides us with a way to construct the exact mapping for any spectral density by simply looking for the family of orthogonal polynomials with respect to the measure $d\mu(x)=h^2(x)dx$. This can be done analytically for many important cases, as we demonstrate in section \ref{sec;examples1} with the SBM, however even if the weight $h^2(x)$ is a complicated function, families of orthogonal polynomials can be founded by using very stable numerical algorithms such as the ORTHPOL package \cite{ORTHPOL}.
\end{remark}

\begin{definition} We will say that a spectral density $J(\omega)$ belongs to the Szeg\"o class if the measure which it induced by $d\mu(x)=h^2(x)dx$, with $g(x)=gx$, belongs to the Szeg\"o class.
\end{definition}

\begin{corollary} For a Szeg\"o class spectral density $J(\omega)$ the tail of the semi-infinite chain tends to a translational invariant chain with
\begin{equation*}
\lim_{n\rightarrow\infty}\omega_n=2\lim_{n\rightarrow\infty}t_n=\frac{gx_{\mathrm{max}}}{4}
\end{equation*}
\end{corollary}

\begin{remark}Typical examples of Szeg\"o class spectral density are those strictly positive in its domain $J(g(x))>0$ for $x\in[0,x_{\mathrm{max}}]$. These includes a wide range of spectral densities with practical interest, particularly the ones used for the SBM discussed in the following section. \end{remark}

\begin{remark}
The asymptotic properties of the frequencies and couplings for $J(\omega)$ in the Szeg\"o class have an important physical implication for open-quantum systems. Away from the quantum system, the chain becomes effectively translationally invariant and it is a standard exercise in condensed matter physics to diagonalize such a chain in terms of plane waves~\cite{SimonsAltland}. For chains derived from a Szeg\"o class $J(\omega)$, the dispersion $\omega_{k}$ of the eigenstates of the asymptotic region of the chain is $\omega_{k}=\omega_{c}(1-\cos(\pi k))/2$ where the momentum runs from $0$ to $1$ $(\omega_c=g(1))$. This result shows rigorously that the asymptotic region can support excitations over the full spectral range of the original environment, as one would intuitively expect on physical grounds. The translational invariance of the asymptotic region means that excitations of the environment can propagate away from the quantum system without scattering and are effectively ``lost'' irreversibly to the environment over time. One could also consider this asymptotic region of the chain as a secondary environment, in which case the system bath Hamiltonian can be recast as a finite chain in which the final mode is damped by the translationally invariant chain. In this picture the damping of the final member of the chain is universal for all spectral densities in the Szeg\"o class. Such a truncated picture might allow for improved efficiency in the numerical simulation of open-system dynamics, and the detailed physics behind this intuitively appealing picture of relaxation and dephasing in the chain picture will be dealt with elsewhere.
\end{remark}

\begin{remark}
A spectral function that appears in nature, most notably in photonic crystals, but which does not fall into the Szeg\"o class would be a spectral function that contains gaps, for instance
\[
J(\omega)=[J_{1}(\omega)\theta(\omega_{1}-\omega)+J_{2}(\omega)\theta(\omega-\omega_{2})]\theta(\omega_{c}-\omega).
\]
Such a spectral function is zero in the region $\omega_{1}\leq\omega\leq\omega_{2}$. However, if $J_{1,2}(\omega)$ are in the Szeg\"o class over $[0,\omega_{1}]$ and $[\omega_{2},\omega_{c}]$ respectively, then the environment can be considered as two separate chains, for which the theorems above apply for each separate chain.
\end{remark}

\section{Applications}\label{sec;examples1}

\subsection{Continuous Spin-Boson spectral densities}\label{sec;examples11}
As a prototypical example let us consider the spin-boson Hamiltonian which describes the interaction of a TLS with an environment of harmonic oscillators \cite{weiss,leggett}
\begin{equation}
H_{\mathrm{SB}}=H_{\mathrm{loc}}+\int_{0}^{x_\mathrm{max}}dxg(x)a_{x}^{\dagger}a_{x}+\frac{1}{2}\sigma_{3}\int_{0}^{x_\mathrm{max}}dxh(x)(a_{x}+a_{x}^{\dagger})\label{hsb},
\end{equation}
where $a_{x}$ are bosonic operators, and the local Hamiltonian of the spin is given by,
\begin{equation}
H_{\mathrm{loc}}=\frac{1}{2}\eta\sigma_{1}+\frac{1}{2}\epsilon\sigma_{3}.
\end{equation}
where $\sigma_1$ and $\sigma_3$ are the corresponding Pauli matrices.

Firstly we will consider spectral functions bounded by a hard cut-off at an energy $\omega_{c}$, hence the cut-off $x_\mathrm{max}=g^{-1}(\omega_c)$ appearing in the integrals in (\ref{hsb}), which are usually parameterized as \cite{bulla}
\begin{equation}
J(\omega)=2\pi\alpha\omega_{c}^{1-s}\omega^{s}\theta(\omega-\omega_c)\label{jsb1},
\end{equation}
where $\alpha$ is the dimensionless coupling strength of the system-bath interaction and $\theta(\omega-\omega_c)$ denotes the Heaviside step function. In the SBM literature, spectral functions with $s>1$ are referred to as super-ohmic, $s=1$ as Ohmic, and $s<1$ as sub-Ohmic. The sub-ohmic and Ohmic cases are known to show interesting critical behaviour at zero temperaure as a function of $\alpha$ \cite{weiss,leggett}. According to (\ref{Jhg}) and our convention in taking linear dispersion relations, this continuous spectral function is related to the Hamiltonian parameters by,
\begin{eqnarray}
g(x)&=&\omega_{c}x,\\
h(x)&=&\sqrt{2\alpha}\omega_{c}x^{s/2}\label{h}.
\end{eqnarray}

With this choice of $g(x)$ and $h(x)$, $x_\mathrm{max}=1$ and absorbing the common factor $\omega_{c}\sqrt{2\alpha}$ in the system operator $\hat{A}=\omega_{c}\sqrt{2\alpha}\sigma_3$, the following matrix elements generate the mapping onto the chain,
\begin{equation}
U_{n}(x)=x^{s/2}\tilde{P}_{n}^{(0,s)}(x),
\end{equation}
here $\tilde{P}_{n}^{(0,s)}(x)=P_{n}^{(0,s)}(x)N_{n}^{-1}$ is a (normalized) shifted Jacobi polynomial (from support [-1,1] to support [0,1]) for which the orthogonality condition reads
\begin{eqnarray}
\int_{0}^{1}dxx^{s}P_{n}^{(0,s)}(x)P_{m}^{(0,s)}(x)&=&\delta_{nm}N_{n}^{2}\\
N_{n}^{2}&=&\frac{1}{2n+s+1}.
\end{eqnarray}
Therefore we get
\[
\hat{A}\int_{0}^{1}dxh(x)(a_{x}+a_{x}^{\dagger})=\frac{1}{2}\sigma_{3}\sqrt{\frac{\eta_{0}}{\pi}}(b_{0}+b_{0}^{\dagger})
\]
where $\eta_{0}=\int_{0}^{\omega_{c}}J(\omega)\,d\omega$ as defined by Bulla {\it et al.} \cite{bulla} in their NRG study of the SBM. They also use a mapping of the SBM onto a chain, however in their case the mapping is performed on an approximate discretised representation of the oscillator reservoir - see section \ref{sec;discrete}.

For the transformed of $\int_{0}^{x_\mathrm{max}}dxg(x)a_{x}^{\dagger}a_{x}$ we obtain from the recursion coefficients of the Jacobi polynomials \cite{Handbook}, frequencies and tunneling amplitudes to be
\begin{eqnarray}
\omega_{n}&=&\frac{\omega_{c}}{2}\left(1 + \frac{s^{2}}{(s + 2n) (2 + s + 2 n)}\right)\label{exenergies},\\
\nonumber\\
t_{n}&=&\frac{\omega_{c}(1 + n) (1 + s + n) }{(s + 2 +2n) (3 + s + 2n)}\sqrt{\frac{3 + s + 2n }{ 1 + s + 2 n}}.\label{extn}
\end{eqnarray}

It is quite easy to check that the measure $d\mu(x)=x^{s}dx$ belongs to the Szeg\"o class for any $s<\infty$, the integral (\ref{SzegoCondition}) on $[0,1]$ is
\[
\int_0^1\frac{\ln(x^s)dx}{\sqrt{1-(2x-1)^2}}=-\pi s\ln(2).
\]
Thus $\lim_{n\rightarrow\infty}\alpha_n=\frac{1}{2}$ and $\lim_{n\rightarrow\infty}\beta'_n=\frac{m^2}{4}=\frac{1}{16}$, so:
\[
\lim_{n\rightarrow\infty}\omega_{n}=\frac{\omega_c}{2}, \quad \lim_{n\rightarrow\infty}t_{n}=\frac{\omega_c}{4},
\]
as it can be seen directly from the expressions (\ref{exenergies}) and (\ref{extn}). Interestingly, smaller values of $s$ lead to faster convergence of $\omega_{n},t_{n}$ towards these limits as $n$ increases.

If instead of a hard cut-off we use an exponential cut-off for the spectral function
\[
J(\omega)=2\pi\alpha\omega_{c}^{1-s}\omega^{s}e^{-\frac{\omega}{\omega_{c}}},
\]
and an integral on $[0,\infty)$, then for $g(x)=\omega_{c}x$ we find
\[
h(x)=\sqrt{2\alpha}\omega_cx^{s/2}e^{-s/2},
\]
the transformation to the chain can achieved using
\begin{equation}
U_{n}(x)=x^{s/2}\tilde{L}_{n}^s(x)=x^{s/2}e^{-s/2}L_{n}^s(x)N_{n}^{-1},
\end{equation}
where $L_{n}^s(x)$ is an associated Laguerre polynomial and $N_{n}$ is the normalisation of the orthogonality relation for the associated Laguerre polynomials:
\begin{eqnarray}
\int_{0}^{1}dxx^{s}e^{-x}L_{n}^s(x)L_{m}^s(x)&=&\delta_{nm}N_{n}^{2}\\
N_{n}^{2}&=&\frac{\mathrm{\Gamma}(n+s+1)}{n!}.
\end{eqnarray}
So the transformation reads
\[
\hat{A}\int_{0}^{\infty}dxh(x)(a_{x}+a_{x}^{\dagger})=\frac{1}{2}\sigma_{3}c_0(b_{0}+b_{0}^{\dagger}),
\]
with $c_0=\omega_{c}\sqrt{2\alpha\mathrm{\Gamma}(s+1)}=\sqrt{\eta_0}$ where $\eta_{0}=\int_{0}^{\infty}J(\omega)\,d\omega$; and for the other term from the recurrence relations of the associated Laguerre polynomials \cite{Handbook}, the frequencies and tunneling are
\begin{eqnarray}
\omega_{n}&=&\omega_{c}\left(2n+1+s\right)\label{Lexenergies},\\
\nonumber\\
t_{n}&=&\omega_{c}\sqrt{(n + 1)(n + s + 1)}.\label{Lextn}.
\end{eqnarray}

\subsection{Discrete spectral densities - The logarithmically discretized Spin-Boson model}\label{sec;discrete}
Our main result (Theorem \ref{mainResult}) maps continuous environments onto discrete semi-infinite chains which are are ideal for t-DMRG simulation. However in some situations we may want to simulate systems which are coupled to a discrete set of oscillators which are not described by a continuous spectral function. This may occur physically in many ways, for example systems coupled to vibrations of small or finite-sized environments are coupled to discrete spectral densities. Many of these discrete environments can be still be transformed into a chain using the same procedure we have presented for the continuous mapping, the only different is that one now makes use of the classical discrete orthogonal polynomials \cite{baik}. One example of this is the finite-size discretisation of the spectral densities considered in section \ref{sec;examples11}. It will be shown in section \ref{sec;hahn} that these discrete environments can be mapped onto chains suitable for t-DMRG analysis using the classical Hahn polynomials. In this subsection we wish to focus on another important class of discrete spectral densities which are logarithimcally discretised approximations to continuous environments. These discretised spectral densities are a key part of the powerful NRG algorithm used to study condensed matter systems, and a detailed description of this numerical method and the reasons for performing this discretisation are given in Ref. \cite{bulla}.

In the NRG approach to the same SBM considered in section \ref{sec;examples1}, the continuous spin-boson Hamiltonian is approximated by representing all the modes in a energy range $\omega_{c}\Delta^{-n}-\omega_{c}\Delta^{-n-1}$ by a single effective degree of freedom described by creation and annihilation operators $a_{n},a_{n}^{\dagger}$. These modes may either be fermionic or bosonic. The logarithmic discretisation parameter $\Delta$ is always greater than unity, and the original continuum spin-boson Hamiltonian is recovered by sending $\Delta\rightarrow 1$. Assuming a spectral density of the form $J(\omega)=2\pi\alpha\omega_{c}^{1-s}\omega^{s}\theta(\omega_{c}-\omega)$, this logarithmic discretisation procedure generates a discrete spin-boson Hamiltonian $H_{L}$ of the form \cite{bulla},
\begin{equation}
H_{L}=H_{loc}+\frac{\sigma_{z}}{2\sqrt{\pi}}\sum_{n=0}^{\infty}\gamma_{n}(a_{n}+a_{n}^{\dagger})+\sum_{n=0}^{\infty}\zeta_{n}a_{n}^{\dagger}a_{n}\label{logsb},
\end{equation}
where
\begin{eqnarray}
\gamma_{n}^{2}&=&\frac{2\pi\alpha}{1+s}\omega_{c}^{2}(1-\Delta^{-(1+s)})\Delta^{-n(1+s)}=\gamma_{s}\Delta^{-n(1+s)},\\
\zeta_{n}&=&\frac{s+1}{s+2}\frac{1-\Delta^{-(s+2)}}{1-\Delta^{-(s+1)}}\omega_{c}\Delta^{-n}=\zeta_{s}\Delta^{-n}.
\end{eqnarray}
As described in detail in Ref. \cite{bulla}, the NRG algorithm can only be applied after the Hamiltonian $H_{L}$ of Eq. (\ref{logsb}) is mapped onto a $1-D$ nearest-neighbour chain of the form,
\begin{equation}
H_{c}=H_{loc}+\frac{1}{2}\sqrt{\frac{\eta_{0}}{\pi}} \sigma_{z} (b_{0}+b_{0}^{\dagger})+\sum_{n}\omega_{n}b_{n}^{\dagger}b_{n}+t_{n}b_{n+1}^{\dagger}b_{n}+t_{n}b_{n}^{\dagger}b_{n+1}\label{hcl}.
\end{equation}
To do this Bulla \textit{et al.} \cite{bulla} first introduce a new set of bosonic modes $b_{n},b_{n}^{\dagger}$ defined by an orthogonal transformation $U$ which acts on the $a_{n},a_{n}^{\dagger}$,
\begin{eqnarray}
b_{n}^{\dagger}&=&\sum_{m}U_{nm}a_{m}^{\dagger},\label{bnl}\\
b_{n}&=&\sum_{m}U_{nm}a_{m},\label{bnldagger}.
\end{eqnarray}
They then derive the following recurrence relation for the matrix elements $U_{nm}$ that bring the Hamiltonian of (\ref{logsb}) to the chain form of (\ref{hcl}) \cite{bulla},
\begin{equation}
\zeta_{n}U_{mn}=\omega_{m}U_{mn}+t_{m}U_{m+1n}+t_{m-1}U_{m-1n}.\label{rec}
\end{equation}
This recurrence relation is then solved numerically by an iterative procedure to provide the hopping parameters $t_{n}$ and site energies $\omega_{n}$ which are the input parameters for the NRG algorithm. This recurrence relation is the standard method for finding the chain parameters  for almost all numerical simulation methods that use the chain representation for impurity problems at present. However the numerical solutions of the recurrence relations are highly sensitive to numerical noise and are often unstable as the dimension of $U$ becomes large. In Ref. \cite{bulla} it is stated that numerical instabilities typically appear after about $25-30$ iterations. The results presented in this paper show that the problem of implementing this chain mapping effectively amount to the problem of finding a suitable basis of orthogonal polynomials to represent the bath degrees of freedom. The recursive method given by Bulla {\it et al.} \cite{bulla} is in fact equivalent to the recursive Stieltjes method for determining the coefficients of orthogonal polynomials, a procedure which is well-known to be numerically unstable for many weight functions, and in certain cases is so inaccurate that it completely fails for all practical purposes \cite{gander}. As orthogonal polynomials are very useful and widely used in numerical quadrature, the problem of determining recurrence coefficients accurately for arbitrary weight functions has been investigated in detail, and alternative numerical algorithms for determining recurrence coefficients to arbitrary accuracy have been developed, most notably the ORTHPOL package of Gautschi \cite{ORTHPOL}. Our results showing that the chain mapping can always be cast as a problem of finding recurrence coefficients of orthogonal polynomials has thus allowed us to make use of these methods, therefore removing the issue of numerical stability from the general problem.

However, for the specific example considered in this section we can in fact solve the recursion relation analytically, and can give explicit formulae for the parameters $t_{n},\omega_{n}$. We now demonstrate that the solution $U_{nm}$ to the recurrence relation of \cite{bulla} can be written as,
\begin{eqnarray}
U_{nm}&=&\frac{\Delta^{-\frac{m(1+s)}{2}}p_{n}(\Delta^{-m},\Delta^{-s},1|\Delta^{-1})}{N_{n}},\label{jac}
\end{eqnarray}
where the functions $p_{n}(\Delta^{-m},\Delta^{-s},1|\Delta^{-1})$ in Eq. (\ref{jac}) are known as the little-q Jacobi polynomials \cite{Koekoek}. These polynomials obey the following orthogonality relation which defines the normalization constants $N_{n}$ that appear in Eq. (\ref{jac}),
\begin{eqnarray}
\delta_{nm}N^{2}_{n}&=&\sum_{k=0}^{\infty}\Delta^{-k(1+s)}p_{n}(\Delta^{-k},\Delta^{-s},1|\Delta^{-1})p_{m}(\Delta^{-k},\Delta^{-s},1|\Delta^{-1})\\
N^{2}_{n}&=&\frac{\Delta^{-n(1+s)}(\Delta^{-1};\Delta^{-1})^{2}_{n}}{(\Delta^{-(s+1)};\Delta^{-1})_{n}^{2}(1-\Delta^{-(2n+1+s)})}\label{ortho1}.
\end{eqnarray}
In Eq. (\ref{ortho1}) $(a:q)_{n}$ is the q-shifted factorial described in \cite{Koekoek}, and defined as

\begin{equation}
(a:q)_{n}=(1-a)(1-aq)(1-aq^2)...(1-aq^{n-1}).
\end{equation}

Noting that $p_{0}(x,a,b|q)=1$ \cite{Koekoek},
\begin{eqnarray}
\sum_{k=0}^{\infty}\Delta^{-k(1+s)}p_{n}(\Delta^{-k},\Delta^{-s},1|\Delta^{-1})&=&\delta_{0n}N^{2}_{0}\\
&=&\delta_{0n}\frac{1}{1-\Delta^{-(s+1)}}.\label{ortho2}
\end{eqnarray}
One further property of the little-q Jacobi polynomials that we require to perform the chain mapping is their three-term recurrence relation. This is given by,
\begin{eqnarray}
\Delta^{-n}p_{j}(\Delta^{-n},\Delta^{-s},1|\Delta^{-1})&=&(A_{j}+C_{j})p_{j}(\Delta^{-n},\Delta^{-s},1|\Delta^{-1})\nonumber\\
&-&A_{j}p_{j+1}(\Delta^{-n},\Delta^{-s},1|\Delta^{-1})\nonumber\\
&-&C_{j}p_{j-1}(\Delta^{-n},\Delta^{-s},1|\Delta^{-1}),
\end{eqnarray}
where,
\begin{eqnarray}
A_{j}&=&\Delta^{-j}\frac{(1-\Delta^{-(j+1+s)})^2}{(1-\Delta^{-(2j+1+s)})(1-\Delta^{-(2j+2+s)})},\\
C_{j}&=&\Delta^{-(j+s)}\frac{(1-\Delta^{-j})^2}{(1-\Delta^{-(2j+s)})(1-\Delta^{-(2j+1+s)})}.
\end{eqnarray}
Using the orthogonality relations one can easily show that the transformation $U$ is an orthogonal matrix and that the $b_{n},b_{n}^{\dagger}$ operators obey the same bosonic or fermionic commutation relations as the $a_{n},a^{\dagger}_{n}$ operators.

Now we express the original Hamiltonian in terms of the $b_{n},b_{n}^{\dagger}$ operators. The second term of (\ref{logsb}) transforms as,
\begin{eqnarray*}
\sum_{n=0}^{\infty}\gamma_{n}(a_{n}+a_{n}^{\dagger})&=&\sum_{n=0}^{\infty}\sum_{k=0}^{\infty}\gamma_{n}U_{kn}(b_{k}+b_{k}^{\dagger})\\
&=&\gamma_{s}\sum_{k=0}^{\infty}N_{k}^{-1}(b_{k}+b_{k}^{\dagger})\sum_{n=0}^{\infty}\Delta^{-n(1+s)}p_{k}(\Delta^{-n},\Delta^{-s},1|\Delta^{-1})\\
&=&\gamma_{s}\sum_{k=0}^{\infty}N_{k}^{-1}(b_{k}+b_{k}^{\dagger})\delta_{0k}N_{k}^{2}\\
&=&\gamma_{s}N_{0}(b_{0}+b_{0}^{\dagger})=\sqrt{\eta}(b_{0}+b_{0}^{\dagger})
\end{eqnarray*}
where we have used property of Eq. (\ref{ortho2}) and in the last step the explicit form of $\gamma_{s}N_{0}$ has been used,
\begin{equation*}
\gamma_{s}^{2}N_{0}^{2}=\frac{2\pi\alpha}{1+s}\omega_{c}^{2}=\int_{0}^{\omega_{c}}J(\omega)d\omega=\eta_{0}
\end{equation*}
where $\eta_{0}$ is the coupling defined by Bulla {\it et al.} \cite{bulla}. The chain couplings and energies arise from the transformation of the third term of (\ref{logsb}), which then becomes,
\begin{eqnarray*}
\sum_{n=0}^{\infty}\zeta_{n}a_{n}^{\dagger}a_{n}&=&\sum_{n=0}^{\infty}\sum_{j,k=0}^{\infty}\zeta_{s}\Delta^{-n}b_{j}^{\dagger}b_{k}U_{jn}U_{kn}\\
&=&\zeta_{s}\sum_{j,k=0}^{\infty}\frac{b_{j}^{\dagger}b_{k}}{N_{j}N_{k}}\\
&\times&\sum_{n=0}^{\infty}\Delta^{-n(1+s)}\overbrace{\Delta^{-n}p_{j}(\Delta^{-n},\Delta^{-s},1|\Delta^{-1})}p_{k}(\Delta^{-n},\Delta^{-s},1|\Delta^{-1})\label{energies}.
\end{eqnarray*}
The orthogonality condition cannot be used directly here as the weight function of the sum $\Delta^{-n(s+1)}$ is multiplied by an extra factor of $\Delta^{-n}$. However, we can absorb this factor in the sum using the recurrence relation of the little-q polynomials on the braced term above. Substituting the RHS of the recurrence relation into Eq. (\ref{energies}) allows the orthogonality relations to be used, and ensures that the resulting couplings are nearest-neighbour only. This generates a $1-D$ chain Hamiltonian $H'$ of the form of (\ref{hc}), with energies $\omega_{n}$ and hoppings $t_{n}$ given by,
\begin{eqnarray}
\omega_{n}&=&\zeta_{s}(A_{n}+C_{n}),\label{e}\\
t_{n}&=&-\zeta_{s}\left(\frac{N_{n+1}}{N_{n}}.\right)A_{n}\label{t}
\end{eqnarray}

We now show that our explicit form for $U_{mn}$ in terms of little-q Jacobi polynomials satisfies the recurrence equation of Bulla {\it et al.} Eq. (\ref{rec}). We first use Eqs. (\ref{e}) and (\ref{t}) to express the RHS of (\ref{rec}) in terms of the little-q Jacobi recurrence coefficients $A_{n},C_{n}$. Substituting $\zeta_{n}=\zeta_{s}\Delta^{-n}$ on the LHS as well, one obtains
\begin{eqnarray}
\zeta_{s}\Delta^{-n}U_{mn}&=&\zeta_{s}(A_{m}+C_{m})U_{mn}-\zeta_{s}\left(\frac{N_{m+1}}{N_{m}}\right)A_{m}U_{m+1n}\\
&-&\zeta_{s}\left(\frac{N_{m}}{N_{m-1}}\right)A_{m-1}U_{m-1n}.
\end{eqnarray}
Now substitute for the $U_{mn}$ using (\ref{jac}) we obtain,
\begin{eqnarray*}
\Delta^{-n}p_{m}(\Delta^{-n},\Delta^{-s},1|\Delta^{-1})&=&(A_{m}+C_{m})p_{m}(\Delta^{-n},\Delta^{-s},1|\Delta^{-1})\\
&+&\left(\frac{N_{m}}{N_{m-1}}\right)^{2}A_{m-1}p_{m-1}(\Delta^{-n},\Delta^{-s},1|\Delta^{-1})\\
&+&A_{m}p_{m+1}(\Delta^{-n},\Delta^{-s},1|\Delta^{-1}).\label{recurs}
\end{eqnarray*}

The final additional that result we need to use is that the following relations holds,
\begin{eqnarray}
C_{m}&=&A_{m-1}\left(\frac{N_{m}}{N_{m-1}}\right)^{2} \hspace{1cm} m>1,\\
C_{0}&=&0 \hspace{1cm} m=0,
\end{eqnarray}
which is easily shown from the definitions of $A_{m},C_{m}$. Substituting this into Eq.(\ref{final}), the equation reduces to,
\begin{eqnarray*}
\Delta^{-n}p_{m}(\Delta^{-n},\Delta^{-s},1|\Delta^{-1})&=&(A_{m}+C_{m})p_{m}(\Delta^{-n},\Delta^{-s},1|\Delta^{-1})\\
&+&A_{m}p_{m+1}(\Delta^{-n},\Delta^{-s},1|\Delta^{-1})\\
&+&C_{m}p_{m-1}(\Delta^{-n},\Delta^{-s},1|\Delta^{-1}).\label{final}
\end{eqnarray*}
This is just the recurrence relation for the little-q Jacobi polynomials. Therefore the original $U_{nm}$ of Eq. (\ref{jac}) is the solution of the recurrence relations of Bulla {\it et al.}. The energies and hopping amplitudes are simple functions of $A_{n},C_{n}$ and the following additional relation can also be derived using Eqs. (\ref{e}) and (\ref{t}),
\begin{equation*}
\omega_{n}=-t_{n}\left(\frac{N_{n}}{N_{n+1}}\right)-t_{n-1}\left(\frac{N_{n-1}}{N_{n}}\right).
\end{equation*}

\subsection{Linearly-discretised baths}\label{sec;hahn}
We now consider an environment consisting in a large, but discrete number $N+1$ of environmental degrees of freedom. The Hamiltonian for such a system can be written as
\begin{equation}
H=H_{loc}+\frac{\sigma_{z}}{2\sqrt{\pi}}\sum_{n=0}^{N}\gamma_{n}(a_{n}+a_{n}^{\dagger})+\sum_{n=0}^{N}\zeta_{n}a_{n}^{\dagger}a_{n}\label{hlsb},
\end{equation}
where $\zeta_{n}=\omega_{c}n/(N+1)$, the coupling constants $\gamma_{n}$ are
\begin{equation}
\gamma_{n}^{2}=2\pi\alpha\frac{\omega_{c}^{2}}{(N+1)^{s+1}}\left( \begin{array}{c} s+n \\
n  \end{array}\right),
\end{equation}
and the bosonic operators obey the discrete commutation relation $[a_{n},a_{m}^{\dagger}]=\delta_{nm}$. The expressions for $\zeta_{n}$ and $\gamma_{n}$ arise from a linear discretisation of the continuous spectral function $J(\omega)=2\pi\alpha\omega_{c}^{1-s}\omega^{s}$ with $N+1$ oscillators equally-spaced in $k$ approximating the bath. As $N\rightarrow\infty$ we recover the continuum limit which was mapped using Jacobi polynomials. This discrete Hamiltonian can be mapped to the nearest-neighbour chain using the orthogonal transformation

\begin{equation}
b_{k}=\sum_{n=0}^{N}\left( \begin{array}{c} s+n \\
n  \end{array}\right)^{\frac{1}{2}}\,Q_{k}(n,s)\rho_{k}^{-1}a_{n},
\end{equation}
where $Q_{l}(n,s)=Q_{l}(n,s,0,N)$ is a Hahn polynomial \cite{baik,Koekoek}. The Hahn polynomials are classical discrete orthogonal polynomials and are defined by the discrete inner product
\begin{equation}
\sum_{n=0}^{N}\left( \begin{array}{c} \alpha+n \\
n  \end{array}\right)\left( \begin{array}{c} \beta+N-x \\
N-x  \end{array}\right)Q_{l}(n,\alpha,\beta,N)Q_{k}(n,\alpha,\beta,N)=\rho_{k}^{2}\delta_{lk}.
\end{equation}
Introducing the Pochhammer symbol $(a)_{k}=a(a+1)...(a+N-1)$, the normalization factors $\rho_{k}$ are given by
\begin{equation}
\rho_{k}=\frac{(-1)^{k}(k+\alpha+\beta+1)_{N+1}(\beta+1)_{k}k!}{(2k+\alpha+\beta+1)(\alpha+1)_{k}(-N)_{k}N!}.
\end{equation}
Using the inner product relation is it easy to show that the new modes $b_{k},b_{k}^\dagger$ obey the same commutation relations as the original modes, and using the identity \cite{Koekoek},
\begin{equation}
\left( \begin{array}{c} s+x \\
x  \end{array}\right)\left( \begin{array}{c} s+y \\
y  \end{array}\right)\sum_{k=0}^{N}\frac{Q_{k}(x,\alpha,\beta,N)Q_{k}(y,\alpha,\beta,N)}{\rho_{k}^{2}}=\delta_{xy},
\end{equation}
one can show that the inverse transform is
\begin{equation}
a_{n}=\sum_{k=0}^{N}\left( \begin{array}{c} s+n \\
n  \end{array}\right)^{\frac{1}{2}}Q_{k}(n,s)b_{k}.\label{ahahn}
\end{equation}

The Hahn polynomials also obey the three-term recurrence relation
\begin{eqnarray}
xQ_{n}(x,\alpha,\beta,N)&=&A_{n}Q_{n+1}(x,\alpha,\beta,N)-(A_{n}+C_{n})Q_{n}(x,\alpha,\beta,N)\nonumber\\
&+&C_{n}Q_{n-1}(x,\alpha,\beta,N),
\end{eqnarray}
where
\begin{eqnarray}
A_{n}&=&\frac{(n+\alpha+\beta+1)(n+\alpha+1)(N-n)}{(2n+\alpha+\beta)(2n+\alpha+\beta+1)},\\
C_{n}&=&\frac{n(n+\alpha+\beta+N+1)(n+\beta)}{(2n+\alpha+\beta)(2n+\alpha+\beta+1)}.
\end{eqnarray}
We now substitute (\ref{ahahn}) into $H$ to obtain the chain Hamiltonian. As in the previous examples, the three-term recurrence and orthogonality of the $Q_{k}(n,s)$ ensures the generation of a discrete nearest-neighour chain with energies and couplings determined by the recurrence coefficients of the Hahn polynomials. The transformed Hamiltonian $H_{c}$ is given by,

\begin{eqnarray}
H_{c}&=&H_{loc}+\frac{1}{2}\sqrt{\frac{\eta}{\pi}}\sigma_{z}(b_{0}+b_{0}^{\dagger})+\sum_{n=0}^{N}\omega_{n}b_{n}^{\dagger}+t_{n}b_{n}^{\dagger}b_{n+1}+t_{n}b_{n+1}^{\dagger}b_{n},\\
\omega_{n}&=&\frac{\omega_{c}}{N+1}(A_{n}+C_{n}),\\
t_{n}&=&\frac{\omega_{c}}{N+1}\frac{\rho_{n+1}}{\rho_{n}}A_{n}.
\end{eqnarray}

The discrete weight function is not in the Szeg\"o class, and at finite $N$ we find that $\omega_{n}\rightarrow \omega_{c}/4,t_{n}\rightarrow 0$ as $n \rightarrow N$. However, taking the limit $N\rightarrow\infty$ leads to the recovery of the results obtained for a continuous power-law spectral density in section \ref{sec;examples11}. Formally this results from the limit
relation~\cite{Koekoek},$$\lim_{N\rightarrow\infty}Q_{n}(Nx,a,b,N)=(-1)^{n}P_{n}^{a,b}(2x-1)/P_{n}^{a,b}(1),$$ although this can also be seen by using taking the limit of the explicit expressions for the recurrence relations $A_{n},C_{n}$ and normalizations $N_{n}$ for the Hahn polynomials. Using the explicit expressions, it is also easy to show that the difference between the chain energies and hoppings for the continuous and linearly discretized bath are of $\mathcal{O}(n/N)$ for $n\ll N$. The differences in the behaviour of the chain parameters in the finite and continuous cases has a physical significance that will be explored in detail elsewhere. Here, it will be enough to say the vanishing of the $t_{n}$ as $n\rightarrow\infty$ is related to the existence of a recurrence time in the dynamics of the discrete problem.

\section{Conclusions} In this article we have demonstrated how the problem of a quantum system coupled linearly to the coordinates of a continuous bosonic or fermionic environment can be mapped exactly onto a one-dimensional model with only short-range interactions. Using the theory of orthogonal polynomials, we have shown how the parameters of this one dimensional representation can be simply obtained from the recurrence coefficients of the orthogonal polynomials of the spectral function describing the environment. Certain important spectral functions correspond to the weight functions of the classical orthogonal polynomials, and the parameters of the chain Hamiltonian can be given in closed analytical form. When the spectral function is not a classical weight function, we can make use of numerical techniques developed for determining recurrence coefficients that allow us to find the chain parameters to high precision and for arbitrary long chains. Our use of orthogonal polynomials simultaneously removes the need to discretise the environmental degrees of freedom and also the numerical instabilities that plague the recursive method previously used to implement this chain transformation. This is particularly important, as hitherto the principal use of the chain transformation is to enable numerical simulation of the dynamics and ground state properties of systems which strongly interact with their environments. The work presented here was principally motivated by an attempt to use the powerful t-DMRG simulation technique to study strongly-interacting open-quantum systems, a problem of considerable topically, especially in the emerging field of quantum effects in biology. The combination of our exact mapping has now made it possible to use t-DMRG to obtain extremely accurate simulations of open-system dynamics with no artefacts arising from having a discretised representation of the environment and a numerical approximation to the mapping. Using the theory of orthogonal polynomials, we have also given examples that could be used in the NRG community.

In addition to these practical advantages, the analytical mapping reveals an interesting structure to the system-environment systems we consider, in particular,the result that the chain representations
of all spectral functions in the {\it physically} broad, Szeg\"o class look the same in the asymptotic limit. This observation may lead to further improvements in the numerical efficiency of open-quantum simulations, allowing the chain representations to be truncated in an efficient way that is universal for all spectral functions in the Szeg\"o class. We have also presented results showing the relations between the chain representations of linearly and logarithmically discretised spin-boson models and the continuous case, results which may also prove useful in determining how to optimize the accuracy of simulations with a discretised representations.

\subsection*{Acknowledgments} We acknowledge Miguel \'Angel Rodr\'iguez and Herbert Spohn for illuminating comments. Financial support from the EU Integrated Project STREP action CORNER and an Alexander von Humboldt Professorship is gratefully acknowledged.

\end{document}